\documentclass[11pt,reqno]{amsart}
\usepackage{amsaddr}
\usepackage[table]{xcolor}

\usepackage[margin=2.7cm]{geometry}
\usepackage{array}
\usepackage{booktabs}
\usepackage[dvipdfmx]{graphicx}
\usepackage[sectionbib]{natbib} 
\usepackage{amsfonts}
\usepackage{amssymb,amsthm}
\usepackage{amsmath}
\usepackage[ruled,vlined,lined,linesnumbered]{algorithm2e}
\usepackage{float}
\usepackage[font=small,labelfont=bf]{caption}

\newtheorem{theorem}{Theorem}[section]

\theoremstyle{definition}

\theoremstyle{remark}

\numberwithin{equation}{section}

\begin{document}

\title{An approximate closed formula for European Mortgage Options }

\author{Manuel Lopez Galvan}

\address{Faculty of Exact and Natural Sciences - University of Buenos Aires -  Argentina  }

\email{manuel.lgalvan@uba.ar, amlopezgalvan@gmail.com , mlopezgalvan@hotmail.com }
\thanks{}


\date{}

\dedicatory{}

\begin{abstract}
The aim of this paper is to investigate the use of close formula approximation for pricing European mortgage options. Under the assumption of logistic duration and normal mortgage rates the underlying price at the option expiry is approximated by shifted lognormal or regular lognormal distribution by matching moments. Once the price function is approximated by lognormal distributions, the option price can be computed directly as an integration of the distribution function over the payoff at the option expiry by using Black-Scholes-Merton close formula. We will see that lower curvature levels correspond to positively skewness price distributions and in this case lognormal approximation leads to close parametric formula representation in terms of all model parameters. The proposed methodologies are tested against Monte Carlo approach under different market and contract parameters and the tests confirmed that the close form approximation have a very good accuracy.     

\end{abstract}

\maketitle
\section{Introduction and Background}
Typically mortgage options are options on TBA pass-through insured against default by Fannie Mae, Freddie Mac, or Ginnie Mae. A TBA is a liquid short-dated forward agreement with a specified coupon, agency, and maturity. A mortgage call gives the holder the right to buy a pass-through at an agreed upon price on the option expiry date and similarly, a mortgage put gives the holder the right to sell the pass-through at a pre-specified price on the expiry date. Since liquidity in the underlying TBA contracts is strictly confined to a few months forward settlement dates, the mortgage options are necessarily short-dated. Mortgage options are commonly available for a broad range of customer specified strikes, maturities, and underlying coupons, and are traded over the counter by most major Wall Street broker dealers. Traditionally, duration is defined as the negative percentage price sensitivity with respect to underlying rate. Usually as mortgage rates fall, prepayments increase and a pass-through’s expected life shortens; as rates rise, prepayments subside and expected cash flows extend. Consequently, the duration of a pass-through shortens as rates fall and lengthens as rates rise suggesting a logistic functional form. In general, the distribution of mortgage prices are negatively skewed and therefore far from lognormal, thus Black-Scholes cannot be expected to produce reasonable mortgage option prices or risk sensitivities. Usually, the common way to make a credible mortgage option pricing model is by Monte Carlo approach. In \cite{Prendergast(2003)} the author explores the pricing and the risk sensitivities of mortgage options by using Monte Carlo simulation with a logistic duration and under normal interest mortgage rates process. In this work we will see that under logistical durations and normal rates assumptions the empirical price distribution could be approximated by negative or positive shifted lognormal distribution and then a modified Black-Scholes formula may be used to price mortgage option. Also, we will show that for small values of curvature the empirical price distribution could be approximated by regular lognormal distribution with parameters depending of logistic duration. In this way, we will see that a power of the price distribution can be expressed as a sum of two lognormal variables then by matching mean and variance we will approximate it by a lognormal distribution. The moment matching methodology has been very useful in finance over the years and one of its first applications was performed in \cite{Levy(1992)}, where the author  approximated the distribution of a basket by matching its first two moments with the moments of a lognormal density function. Once the terminal price distribution is approximated by a lognormal variable the price of an European option on the mortgage can be computed directly by using the Black-Scholes-Merton formula. This pricing approximation will also allow us to approximate the option risk sensitivities by close formulas. We will see that the accuracy of this close approximation formula will depend of the curvature of the duration profile and also depend of the relative position between contract and market parameters.

The duration $D$ is defined as the negative percentage price sensitivity with respect to underlying rate, that is

\begin{eqnarray}\label{priceDiffEquation}
dP(r)=-D(r)P(r)dr
\end{eqnarray}                              

where $P(r)$ is the price of the bond as a function of underlying rate $r$. The duration of a mortgage varies across rates
monotonically, rising from a minimal level for low rates to approach a stable level for high rates. We parameterize the duration by means of a logistic function of the relative rate r, 

\begin{eqnarray}
D(r)= L+ \frac{U}{1+e^{-C(r-x_0)}}
\end{eqnarray}                              

where $C$ is the curvature, $L$ is the lower bound, $U$ is the upper bound and $x_0$ is the coupon. 
Integrating Equation \ref{priceDiffEquation} results in an expression for price as a function of the relative rate,

\begin{eqnarray}\label{priceEquation}
P(r)=k e^{-Lr}\big(1+e^{C(r-x_0)}\big)^{-U/C}
\end{eqnarray}
    
where $k$ is a constant level that may be deduced by matching the market price quote with the model price \ref{priceEquation} at the current mortgage rate. Thus, if $P_0$ and $r_0$ denotes the current spot market price and the current mortgage rate respectively, then the constant level is, $$k=P_0e^{Lr_0}\big(1+e^{C(r_0-x_0)}\big)^{U/C}$$  
We assume that the mortgage rate follows a normal process with drift,

\begin{eqnarray}
dr_t&=&\mu dt + \sigma dW_t \\
r(0)&=&r_0
\end{eqnarray}

where $\mu$ is the drift, $\sigma$ the volatility and $W_t$ is a Brownian motion. The rate distribution at a future time $t=T$ is,

\begin{eqnarray}
r_T\sim \mathcal{N}(r_0 +\mu T, \sigma^2 T)
\end{eqnarray} 
  
Now consider a Call mortgage option struck at a strike level $K$, expiring at $T$ and the risk-free rate of interest $r_f$. Under the risk-neutral world the call price can be computed as the expected value of the payoff at the expiry date discounted at the risk-free rate of interest,that is,

\begin{eqnarray}
\mathcal{C} = e^{-r_f T} E((P(r_T)-K)^+)
\end{eqnarray}
   
where $P(r_T)$ is the mortgage price distribution at the expiring.     


\section{Close pricing approximation}
In this section we will develop the algorithm for pricing mortgage option by using shifted Black-Scholes formulas under shifted lognormal and lognormal price distribution approximation. 
\subsection{Shifted Lognormal approximation}

We start recalling some properties of lognormal and shifted lognormal distribution.
Suppose that $X\sim \mathcal{N}(\mu_X, \sigma^2_X)$ is normal, then the random variable $Z = e^{X}$ is lognormal and is noted as $\mbox{Log}\mathcal{N}(\mu_X, \sigma^2_X)$. It is well known that the first and second moment are $E(Z)=e^{\mu_X + \frac{1}{2}\sigma^2_X}$ and $E(Z^2)=e^{2\mu_X +2\sigma^2_X}$ respectively.

A usefully lognormal property recalls to the power and multiplication, indeed if $Z=e^X \sim \mbox{Log}\mathcal{N}(\mu_X, \sigma^2_X)$ and $a$ is a constant then,

\begin{eqnarray}\label{power_constant_logN_properties}
aZ  &\sim & \mbox{Log}\mathcal{N}(\mu_X + \log(a), \sigma^2_X)\\
Z^a &\sim & \mbox{Log}\mathcal{N}(a\mu_X , a^2\sigma^2_X)
\end{eqnarray} 

The shifted lognormal distribution or three parameter lognormal is a lognormal distribution by addition of a shift parameter $\tau$. Thus, if $Z=e^X$ is a lognormal random variable then $Z_{\tau} = \tau - Z$ has the negative shifted lognormal distribution and $Z_{\tau} = \tau + Z$ has the shifted lognormal distribution, in general we note $Z_{\tau} \sim \mbox{Log}\mathcal{N}(\tau,\mu_X, \sigma^2_X)$.

Given a sample of a distribution $p_1, p_2,...,p_n$, we denote the central sample moments by;

\begin{eqnarray}
\bar{p} &=& \dfrac{1}{n}\sum^n_{i=1} p_i\\
m_2 &=& \dfrac{1}{n}\sum^n_{i=1} (p_i - \bar{p})^2\\
m_3 &=& \dfrac{1}{n}\sum^n_{i=1} (p_i - \bar{p})^3 
\end{eqnarray} 

The usual way to estimate the parameters $\tau,\mu_X, \sigma_X$ is by using the method of moment; thus given a sample $(p_i)$ equating the sample first moment (the sample mean) with its population value (the population mean), and equating the second and third sample central moments with their population values yields:

\begin{eqnarray}
\bar{p} & = & \tau + e^{\mu_X + \frac{1}{2}\sigma^2_X }\\ 
m_2 & =  & e^{2\mu_X + \sigma^2_X}(e^{\sigma^2_X} -1)\\
m_3 & = & e^{3\mu_X + \frac{3}{2}\sigma^2_X}(e^{\sigma^2_X} -1)^2(e^{\sigma^2_X}+2)
\end{eqnarray}

The skewness of a random variable is a measure of the asymmetry of the probability distribution of a real-valued random variable about its mean, thus if $V$ is a random variable, the skewness is defined as;
$$Skew(V) = \frac{E\bigg( \big(V-E(V)\big)^3\bigg)}{{s_V}^{3}} $$ 
  
where $s_V$ is the standard deviation of $V$. Given a sample of the random variable an estimator of the skewness is,

$$\widehat{Skew(V)} = \frac{m_3}{\widehat {s}^3} $$   

The key to the approximation of the price distribution is based on the Newton's binomial generalization to real exponents; by applying Newton's it is possible to approximate the price distribution $P(r)$ as a Basket with lognormal terms. Indeed, 

\begin{eqnarray}
P(r)=k e^{-Lr}\big(1+e^{C(r-x_0)}\big)^{-U/C} & \approx & k  e^{-Lr} \sum^n_{j=0} \binom{-U/C}{j} e^{j C (r - x_0)} \\
 & = & \sum^n_{j=0} \binom{-U/C}{j} k e^{(jC - L)r - jCx_0} 
\end{eqnarray}

where $\binom{-U/C}{j} = \frac{ -U/C (-U/C - 1) ... (-U/C - j + 1)}{j!} $ and $n$ huge.

Usually correlated lognormals sum has no closed-form expression, however it may be reasonably approximated by another shifted lognormal distribution by matching the moments. This basket approximation can be negatively skewed, and therefore direct lognormal distribution cannot be fitted.  A good way to avoid these problems is through the choice of negative shifted lognormal distribution. \cite{Borovkova(2007)} studied the choice of the approximating basket distribution by using a generalized family of lognormal distributions and this approximations copes negative skewness. In this work negative skewness basket  are approximated by negative shifted lognormal distribution. 
\bigskip
 
\subsubsection{Option valuation by using Shifted Lognormal approximation}  
Following a similar approach as in \cite{Borovkova(2007)}, if the skewness of terminal price distribution is negative, then negative shifted lognormal distribution are chosen as an approximating distribution; otherwise if the skewness is positive, then the approximating distribution are the shifted lognormal.  
Thus, by using the shifted lognormal approximation the payoff of a call option on the maturity date $T$ and strike $K$ is approximated by the payoff of a put with strike $\tau - K$ when the skewness of price is negative. Otherwise, when the skewness is positive the call with strike $K$ is approximated with a call with strike $K - \tau$:

\[
\left\{ \begin{array}{lcl}
(P(r_T)- K)^+ \simeq (\tau - Z - K)^+ = \big( (\tau - K) - Z\big)^+ & \mbox{ if } & \widehat{Skew(P(r_T))} <0 \\
& & \\
(P(r_T)- K)^+ \simeq (\tau + Z - K)^+ = \big(Z - (K - \tau)\big)^+ & \mbox{ if } & \widehat{Skew(P(r_T))} >0
\end{array}
\right.
\] 

\bigskip
 
where $Z$ is a lognormal random variable. These argument lead to value the option by using the Black-Scholes formula.  

The algorithm that we have developed for pricing mortgage options is summarized below;

\bigskip

\begin{algorithm}[H]\label{shifted_lognormal_pricing}
\small
\caption{Computes mortgage options by using Shifted Lognormal Price approximation }
\KwData {$r_f, T, K, L, C, U, \mu, \sigma, r_0, x_0, P_0$}
\KwResult {returns the option value approximation $\mathcal{C}_{SLN}$.}

P $\leftarrow$ Generate a sample of the terminal price distribution $P(r_T)$ \\

\If {skew(P) $<0$} {$(\tau,\mu_X, \sigma_X)$ $\leftarrow$ Fit a shifted lognormal distribution on P by matching moments.\\
                        $M_1$ $\leftarrow$ $e^{\mu_X + 0.5\sigma_X^2}$\\
                        $M_2$ $\leftarrow$ $e^{2\mu_X + 2\sigma_X^2}$\\
                        $W$ $\leftarrow$ $\sqrt{\bigg(\log\bigg( \dfrac{M_2}{M_1^2} \bigg) \bigg)}$\\
                        $d_1, d_2$ $\leftarrow$ $\dfrac{\log( M_1 ) - \log( - K - \tau ) \pm 0.5 W^2}{W}$\\
                        $\mathcal{C}_{SLN}$ $\leftarrow$ $e^{-r_fT}\big( (-K-\tau)N(-d_2) - M_1N(-d_1)\big) $ }

\If {skew(P) $>0$} {$(\tau,\mu_X, \sigma_X)$ $\leftarrow$ Fit a shifted lognormal distribution on P by matching moments.\\
                        $M_1$ $\leftarrow$ $e^{\mu_X + 0.5\sigma_X^2}$\\
                        $M_2$ $\leftarrow$ $e^{2\mu_X + 2\sigma_X^2}$\\
                        $W$ $\leftarrow$ $\sqrt{\bigg(\log\bigg( \dfrac{M_2}{M_1^2} \bigg) \bigg)}$\\
                        $d_1, d_2$ $\leftarrow$ $\dfrac{\log( M_1 ) - \log( K - \tau ) \pm 0.5 W^2}{W}$\\
                        $\mathcal{C}_{SLN}$ $\leftarrow$ $e^{-r_fT}\big(M_1 N(d_1) - ( K - \tau )N(d_2) )\big)$ }
\end{algorithm}

\bigskip

\subsection{Lognormal approximation}

In the case of a positively skewed price distribution, it may be reasonably approximated by another lognormal distribution by matching the first two central moment. Let $(e^{X_1},e^{X_2})$ be a correlated lognormal random vector with $X_i \sim \mathcal{N}(\mu_i, \sigma^2_i)$ then the equation to be approximated is 

\begin{eqnarray}\label{sumlognormalaprox}
Y:=e^{X_1}+e^{X_2} \approx Z:=e^X \sim \mbox{Log}\mathcal{N}(\mu_X, \sigma^2_X).
\end{eqnarray}

A commonly used way to approximate \ref{sumlognormalaprox} is to find $\mu_X$ and $\sigma^2_X$ by matching the first two central moment of $Y$ and $Z$. Indeed, recalling \cite{Safak(1994)} since $$E(e^{X_1}e^{X_2})= e^{\mu_1+\mu_2+\frac{1}{2}(\sigma^2_1 + \sigma^2_2+ 2 \mbox{Cov}(X_1,X_2))}$$      
then the second central moment of a sum of lognormal is $$E(Y^2)=E((e^{X_1}+e^{X_2})^2)=e^{2\mu_1+2\sigma^2_1} + 2 e^{\mu_1+\mu_2+\frac{1}{2}(\sigma^2_1 + \sigma^2_2+ 2 \mbox{Cov}(X_1,X_2))}+e^{2\mu_2+2\sigma^2_2}.$$
thus by matching the moments, 

$$
\left\{
\begin{array}{cl}
E(Y)= E(Z) \Rightarrow & e^{\mu_X + \frac{1}{2}\sigma^2_X}=e^{\mu_1 + \frac{1}{2}\sigma^2_1}+e^{\mu_2 + \frac{1}{2}\sigma^2_2}\\
   &\\
E(Y^2)= E(Z^2)\Rightarrow & e^{2\mu_X +2\sigma^2_X} = e^{2\mu_1+2\sigma^2_1} + 2 e^{\mu_1+\mu_2+\frac{1}{2}(\sigma^2_1 + \sigma^2_2+ 2 \mbox{Cov}(X_1,X_2))}+e^{2\mu_2+2\sigma^2_2}
\end{array}
\right.
$$ 

and a straightforward computation by simple substitution allows us to get,

\begin{eqnarray}\label{muXsigmaX}
\mu_X &=& \log \bigg(e^{\mu_1 + \frac{1}{2}\sigma^2_1}+e^{\mu_2 + \frac{1}{2}\sigma^2_2}\bigg) - \dfrac{\sigma^2_X}{2} \\
\sigma^2_X &=& \log \bigg( \dfrac{e^{2\mu_1+2\sigma^2_1} + 2 e^{\mu_1+\mu_2+\frac{1}{2}(\sigma^2_1 + \sigma^2_2+ 2 \mbox{Cov}(X_1,X_2))}+e^{2\mu_2+2\sigma^2_2}}{ \big(e^{\mu_1 + \frac{1}{2}\sigma^2_1}+e^{\mu_2 + \frac{1}{2}\sigma^2_2}\big)^2 }\bigg) 
\end{eqnarray}

\bigskip

\subsubsection{Option valuation by using Lognormal approximation}
Under the assumptions of positively skew of empirical price distribution, then it is possible to approximate it by a lognormal distribution with a close parametric formula in terms of model parameters. The Theorem below shows the lognormal approximation, the option price approximation and option delta approximation in terms of model parameters.         

\begin{theorem}\label{parametriclognormal} Suppose that the mortgage rate follows a normal process with parameters $\mu, \ \sigma$ and assume that the power of terminal price distribution $P(r_T)^{-C/U}$ has positively skewness then there exits parameters $\mu_{P_{r_T}}$ and $\sigma_{P_{r_T}}$ in terms of model parameters $\mu,\sigma,U,L,C,T,k$ such that,
\begin{enumerate}

\item $P(r_T)\approx \mbox{Log}\mathcal{N}(\mu_{P_{r_T}},\sigma^2_{P_{r_T}})$, that is the terminal price distribution could be approximated by a lognormal distribution with parametric parameters $\mu_{P_{r_T}}$ and $\sigma_{P_{r_T}}$.
\bigskip

\item $\mathcal{C}\sim \mathcal{C}_{LN} := e^{-r_f T} \big( e^{\mu_{P_{r_T}}+\frac{1}{2}\sigma^2_{P_{r_T}}}N(d_1) - KN(d_2)\big)$ where, 
$$d_1 = \dfrac{\log\bigg(\dfrac{e^{\mu_{P_{r_T}}+\frac{1}{2}\sigma^2_{P_{r_T}}}}{K}\bigg) + \dfrac{\sigma^2_{P_{r_T}}}{2}}{\sigma_{P_{r_T}}}  \  \ \mbox{and} \ \ d_2 = \dfrac{\log\bigg(\dfrac{e^{\mu_{P_{r_T}}+\frac{1}{2}\sigma^2_{P_{r_T}}}}{K}\bigg) - \dfrac{\sigma^2_{P_{r_T}}}{2}}{\sigma_{P_{r_T}}}$$ and $N$ denotes the cumulative distribution function of standard normal distribution.
\bigskip

\item Delta and Gamma $\Delta = \dfrac{\partial \mathcal{C}}{\partial P_0}, \Gamma = \dfrac{\partial^2 \mathcal{C}}{\partial^2 P_0} $ could be approximated by the closed parametric formulas given by 

$$
\Delta_{LN}  := e^{-r_f T} \dfrac{e^{\mu_{P_{r_T}}+\frac{1}{2}\sigma^2_{P_{r_T}}}N(d_1)}{P_0}, \  
\Gamma_{LN}  :=  e^{-r_f T} \dfrac{e^{\mu_{P_{r_T}}+\frac{1}{2}\sigma^2_{P_{r_T}}}\phi(d_1)}{P_0^2 \sigma_{P_{r_T}} }$$

where $\phi(x)=\frac{\partial N(x)}{\partial x} = \frac{1}{\sqrt{2 \pi}}e^{-\frac{x^2}{2}}$.

\end{enumerate}
\end{theorem}

\begin{proof} 
\begin{enumerate}
\item From the given Price equation solution as function of rate, an algebraic straight forward computation shows that the terminal price distribution may be expressed as,

\begin{eqnarray}\label{Pr_alt_expresion}
P(r_T)^{-\frac{C}{U}} = k^{-\frac{C}{U}} \left[\bigg({\dfrac{P(r_T)}{k}}\bigg)^{-\frac{C}{U}}\right] = k^{-\frac{C}{U}} \big( e^{\frac{LC}{U}r_T} + e^{C(\frac{L}{U}+1)r_T - Cx_0}\big) 
\end{eqnarray}  
  
Calling $X_1= \frac{LC}{U}r_T$ and $X_2=C\big(\frac{L}{U}+1\big)r_T - Cx_0$ and since $r_T$ is normally distributed with mean and variance given by $r_0+\mu T$ and $(\sigma \sqrt{T})^2$, then it follows that $X_1\sim \mathcal{N}(\mu_1,\sigma^2_1)$ and $X_2\sim \mathcal{N}(\mu_2,\sigma^2_2)$ are normally distributed where,

\begin{align}
\mu_1 &= \frac{LC}{U}\big(r_0 + \mu T\big) & \sigma^2_1 &= \bigg(\frac{LC}{U}\sigma \sqrt{T}\bigg)^2 \\
\mu_2 &= C\bigg(\frac{L}{U}+1\bigg)\big(r_0+\mu T\big)-Cx_0 & \sigma^2_2 &= \bigg(C\bigg(\frac{L}{U}+1\bigg)\sigma \sqrt{T}\bigg)^2
\end{align}
The covariance between the random variables $X_1$ and $X_2$ may also be expressed in terms of $L,U,C, \sigma$ and $T$, indeed,

\begin{equation}
\mbox{Cov}(X_1,X_2)=\dfrac{LC}{U}C\bigg(\dfrac{L}{U}+1\bigg)\mbox{Var}(r_T)=\dfrac{LC}{U}C\bigg(\dfrac{L}{U}+1\bigg)\sigma^2 T
\end{equation} 

From Equation \ref{Pr_alt_expresion} it follows that $$Skew\big(P(r_T)^{-\frac{C}{U}}\big)=sign(k^{-\frac{C}{U}})Skew\big(e^{\frac{LC}{U}r_T} + e^{C(\frac{L}{U}+1)r_T - Cx_0}\big)$$ thus, the terminal price distribution may be expressed as a power of a positively skewness sum of lognormal variables times a constant. Then by using the above constructed parameters $\mu_1, \mu_2, \sigma^2_1,\sigma^2_2$ and the lognormal sum matching parameters $\mu_X$ and $\sigma^2_X$ given by Equation \ref{muXsigmaX} the terminal price distribution may be approximated by,  

\begin{eqnarray}
P(r_T)\approx k \big(e^X\big)^{-\frac{U}{C}} \ \ \mbox{where} \  X \sim \mathcal{N}(\mu_X,\sigma^2_X)
\end{eqnarray}

Then by using the power and multiplication properties given by Equation \ref{power_constant_logN_properties} the final price distribution is approximated by a lognormal distribution with parameters,

\begin{eqnarray}
P(r_T)\approx \mbox{Log}\mathcal{N}(\mu_{P_{r_T}},\sigma^2_{P_{r_T}})  \ \mbox{where} \
\mu_{P_{r_T}} = -\frac{U}{C}\mu_X + \log(k) \ ,  \ \sigma_{P_{r_T}} = \frac{U}{C}\sigma_X  
\end{eqnarray}  

\item Since $P(r_T)$ is approximated by a lognormal distribution with parameters $\mu_{P_{r_T}}$ and $\sigma^2_{P_{r_T}}$, the call option price can be valued by the general Black-Scholes-Merton expectation formula (see \cite{Hull(2006)} for a proof); thus by using $E(P(r_T)) = e^{\mu_{P_{r_T}}+\frac{1}{2}\sigma^2_{P_{r_T}}}$ we have,

\begin{eqnarray}\label{call_close_formula_aprox}
\mathcal{C} &\sim & e^{-r_f T} \bigg( E(P(r_T))N(d_1) - K N(d_2)\bigg) =  e^{-r_f T}\big(e^{\mu_{P_{r_T}}+\frac{1}{2}\sigma^2_{P_{r_T}}}N(d_1) - KN(d_2)\big)
\end{eqnarray}

\begin{eqnarray}
d_1 &=&  \dfrac{\log\bigg(\dfrac{E(P(r_T))}{K}\bigg) + \dfrac{\sigma^2_{P_{r_T}}}{2}}{\sigma_{P_{r_T}}} =
 \dfrac{\log\bigg(\dfrac{e^{\mu_{P_{r_T}}+\frac{1}{2}\sigma^2_{P_{r_T}}}}{K}\bigg) + \dfrac{\sigma^2_{P_{r_T}}}{2}}{\sigma_{P_{r_T}}}\nonumber \\ 
d_2 &=&  \dfrac{\log\bigg(\dfrac{E(P(r_T))}{K}\bigg) - \dfrac{\sigma^2_{P_{r_T}}}{2}}{\sigma_{P_{r_T}}} =
 \dfrac{\log\bigg(\dfrac{e^{\mu_{P_{r_T}}+\frac{1}{2}\sigma^2_{P_{r_T}}}}{K}\bigg) - \dfrac{\sigma^2_{P_{r_T}}}{2}}{\sigma_{P_{r_T}}}\nonumber 
\end{eqnarray}

\item The Greeks are followed from a direct derivation of the closed formula pricing approximation. Indeed, recalling the lognormal approximation parameters for $P(r_T)$, 

\begin{eqnarray}
\mu_{P_{r_T}} &=& -\frac{U}{C}\mu_X + \log \bigg(P_0 e^{Lr_0}\big(1+e^{C(r_0-x_0)}\big)^{\frac{U}{C}}\bigg)\nonumber \\
\sigma_{P_{r_T}} &=& \frac{U}{C}\sigma_X \nonumber
\end{eqnarray}

and since $\mu_X$ and $\sigma_X$ are not price depended we have,
\begin{eqnarray}
\dfrac{\partial \big(e^{\mu_{P_{r_T}}+\frac{1}{2}\sigma^2_{P_{r_T}}}\big)} {\partial P_0} &=& \dfrac{ e^{\mu_{P_{r_T}}+\frac{1}{2}\sigma^2_{P_{r_T}}}}{P_0} \nonumber \\
\dfrac{\partial d_1} {\partial P_0} = \dfrac{\partial d_2} {\partial P_0} &=& \dfrac{1}{\sigma_{P_{r_T}}P_0}  \nonumber       
\end{eqnarray}
Using the Formula \ref{call_close_formula_aprox} and applying the above relationship jointly with product and chain rule, we have,

\begin{eqnarray}
\dfrac{\partial \mathcal{C}}{\partial P_0} \sim e^{-r_f T} \bigg(e^{\mu_{P_{r_T}}+\frac{1}{2}\sigma^2_{P_{r_T}}}\frac{N(d_1)}{P_0}+
\dfrac{1}{\sigma_{P_{r_T}}P_0}\big(e^{\mu_{P_{r_T}}+\frac{1}{2}\sigma^2_{P_{r_T}}} \phi(d_1) - K \phi(d_2)\big)\bigg)
\end{eqnarray}
then the result can achieve from the following equality,
$$\log\bigg(\dfrac{\phi(d_1)}{\phi(d_2)}\bigg) = \frac{1}{2}\bigg(d^2_2 - d^2_1\bigg) = - \log \bigg( \dfrac{e^{\mu_{P_{r_T}}+\frac{1}{2}\sigma^2_{P_{r_T}}}}{K} \bigg). $$  
The Gamma approximation is obtained from Delta approximation by directly deriving,
$$  \dfrac{\partial}{\partial P_0}\bigg(\dfrac{\partial \mathcal{C}}{\partial P_0}\bigg) \approx  \dfrac{\partial}{\partial P_0} \bigg(e^{-r_f T} \dfrac{e^{\mu_{P_{r_T}}+\frac{1}{2}\sigma^2_{P_{r_T}}}N(d_1)}{P_0} \bigg).$$   
\end{enumerate}

\end{proof}

\section{Test Results}
In this section we perform test by looking at the effect of varying model and contract parameters such as, strikes, volatilities and curvatures. The test are performed on both of our approach by approximating by shifted lognormal and lognormal distribution. A modified Shapiro-Wilk Goodness-of-Fit Test is used to test the shifted lognormal assumptions by estimating the threshold parameter via the zero-skewness method. We also have compared the empirical price distribution against the approximation distribution by using graphics test such Boxplots and QQPlots. Our standard test example is a mortgage Call option with parameters given by:

\begin{table}[H]
\centering
{\fontfamily{ptm}\selectfont{
\rowcolors{2}{lightgray!20}{} 
\begin{tabular}{*2c}
\toprule
Parameters & Values         \\                
\toprule
         &                     \\
  T (\mbox{Option Expiry}) &    90 \ \mbox{days} (90/360)                      \\
 $r_f$ &   0.0209                              \\ 
  $K$   &    100                \\      
 $P_0$  & 100     \\
 $r_0$ & 0.01    \\
 $\mu$ & 0 \\
 $L$ & 1 \\
 $U$ & 9 \\
 $\sigma$ & 0.02 \\
 $x_0$ & 0.055 \\
\bottomrule                               
\end{tabular}
}}
\caption{Default parameters values. }\label{parametersDefault}
\end{table}

\subsection{Empirical and Shifted lognormal distributions}   
Here we test the assumptions of empirical price distribution on the shifted lognormal. A sample of the empirical price distribution is generated by performing the followings step:
\begin{itemize}
\item Generate a sample of the terminal mortgage rates $r_1, r_2, ...r_n$ of length $n$ by sampling a normal with mean $r_0 + \mu T $ and standard deviation given by $\sigma\sqrt{T}$.
\item Compute $p_1:=P(r_1), p_2:=P(r_2),...,p_n:=P(r_n)$ by using Equation \ref{priceEquation}. 
\end{itemize}    

We test the null hypothesis that the empirical price sample $p_1, p_2,...,p_n$ may come from shifted lognormal with parameters $(\tau,\mu_X, \sigma^2_X)$ against the alternative hypothesis that the sample come from other distribution. The Table \ref{SWShiftedLog} shows the results of the modified Shapiro-Wilk Goodness-of-Fit test.

\begin{table}[H]
\centering
\begin{tabular}{|c|c|c|c|c|c|c|}
  \hline
  C & skew       & W statistic & p-value & $\mu_X$ & $\sigma_X$ & $\tau$ \\ 
  \hline
 0.5 & 0.1236835 & 0.9993015 & 0.5455999 & 4.8383091 & 0.0431144 & -26.3897938 \\ 
 1   & 0.1162353 & 0.9993010 & 0.5451805 & 4.8911724 & 0.0405153 & -33.2424301 \\ 
 2   & 0.1010308 & 0.9992999 & 0.5441955 & 5.0126047 & 0.0352110 & -50.4261475 \\ 
 3   & 0.0854314 & 0.9992987 & 0.5430310 & 5.1612134 & 0.0297708 & -74.5039283 \\ 
 4   &  0.0694575 & 0.9992973 & 0.5417093 & 5.3487801 & 0.0242015 & -110.4771148 \\ 
 5   & 0.0531310 & 0.9992958 & 0.5402571 & 5.5969911 & 0.0185106 & -169.7371020 \\ 
 6   & 0.0364751 & 0.9992942 & 0.5387055 & 5.9530868 & 0.0127061 & -285.0605100 \\ 
 10 & -0.0329438 & 0.9992875 & 0.5297216 & 5.9722431 & 0.0114670 & -492.2735105 \\ 
 15 & -0.1239802 & 0.9992812 & 0.5177176 & 4.5395958 & 0.0430598 & -193.5528287 \\ 
 20 & -0.2164493 & 0.9992826 & 0.5163660 & 3.8734434 & 0.0748392 & -148.0180305 \\ 
 30 & -0.3918973 & 0.9993261 & 0.5729071 & 3.0693251 & 0.1329803 & -121.4553139 \\ 
 40 & -0.5349625 & 0.9994000 & 0.6806394 & 2.5738760 & 0.1755753 & -113.0580899 \\ 
   \hline
\end{tabular}
\caption{Modified Shapiro-Wilk Goodness-of-Fit test, Skew and Fitted parameters under different curvature levels.}\label{SWShiftedLog} 
\end{table}

Under different curvature levels, the Shapiro-Wilk does not reject the null hypothesis that the empirical price is a shifted lognormal distribution. In addition, in order to compare the empirical distribution against the shifted lognormal we perform graphical test such as qqplot between them. Thus, a sample of the approximated shifted lognormal distribution is generated by performing:

\begin{itemize}
\item Fit a shifted lognormal distribution $Z_{\tau} \sim \mbox{Log}\mathcal{N}(\tau,\mu_X, \sigma^2_X)$ on $p_1, p_2,...,p_n$ by matching moments.
\item Generate a sample of $Z_{\tau}$ of length $n$. 
\end{itemize}
     
Figure \ref{distribution} illustrates the performance of shifted lognormal approximation against empirical price distribution under different curvature levels. 

\begin{figure}[H]  
\begin{center}
\includegraphics[scale=0.95]{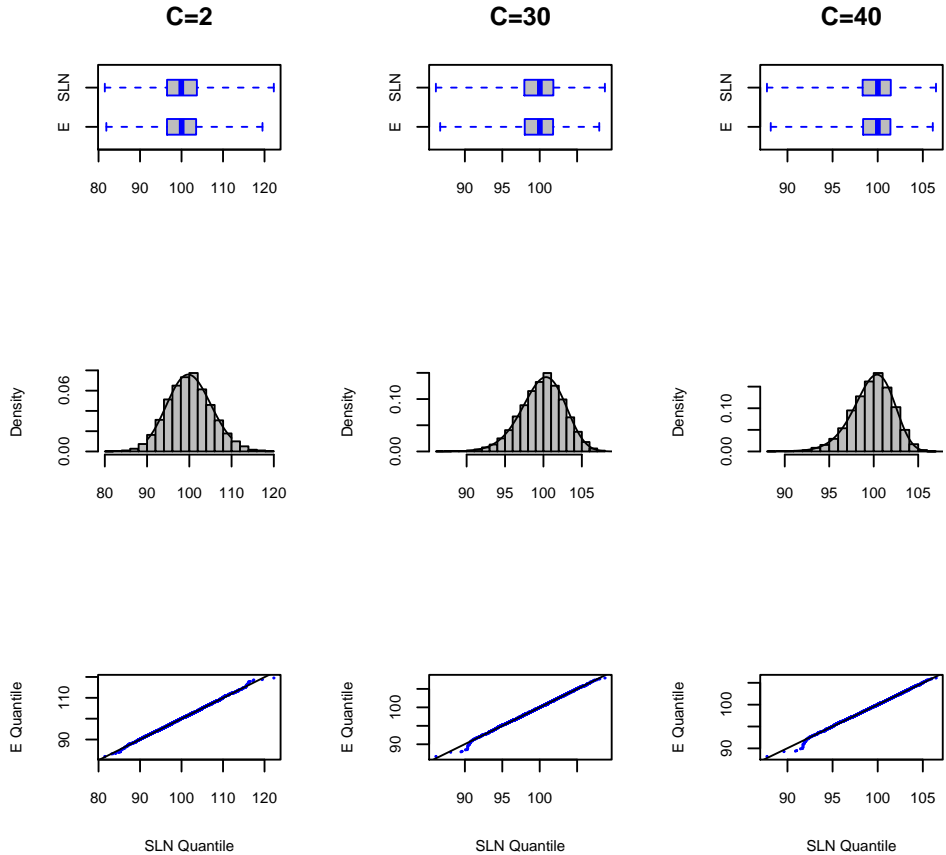}
\end{center}
\caption{Comparison between Empirical price distribution (E) and Fitted Shifted Lognormal (SLN) distribution for different curvatures levels.}\label{distribution}
\end{figure} 

The qqplot indicates a very good alignment between the quantiles and the empirical histogram shows with similar shapes against the parametric shifted lognormal density. It seems the effect of curvature against skew; small curvature correspond to a gradual change of the duration and therefore implicating positive skew empirical distribution; otherwise higher curvature levels correspond to negative skew.    

\subsection{Option price Sensitivity on Shifted Lognormal approximation}
In this section we perform the pricing methodology given by Algorithm \ref{shifted_lognormal_pricing}. In order to assess the performance of Shifted Lognormal approximation, here we perform sensitivity test by varying strike and volatility against different curvature levels. As reference values we compute the option prices by a Monte Carlo simulations and the number of simulations was chosen large enough, $n=70000$. The strike $K$ is varied from 97 to 103, curvature is varied from 0.5 to 40, rate volatility is varied from 0.005 to 0.04 and all other parameters set to default values. Figure \ref{PriceSensitivity1} shows the effect of changing strikes and volatility against curvature and the heat map shows the relative difference respect to Monte Carlo pricing thus, if $\mathcal{C}_{MC}$ denotes the Monte Carlo call price then the heat map illustrates $\frac{(\mathcal{C}_{SLN} - \mathcal{C}_{MC})100}{\mathcal{C}_{MC}}$.         

\begin{figure}[H]  
\begin{center}
\includegraphics[scale=1.1]{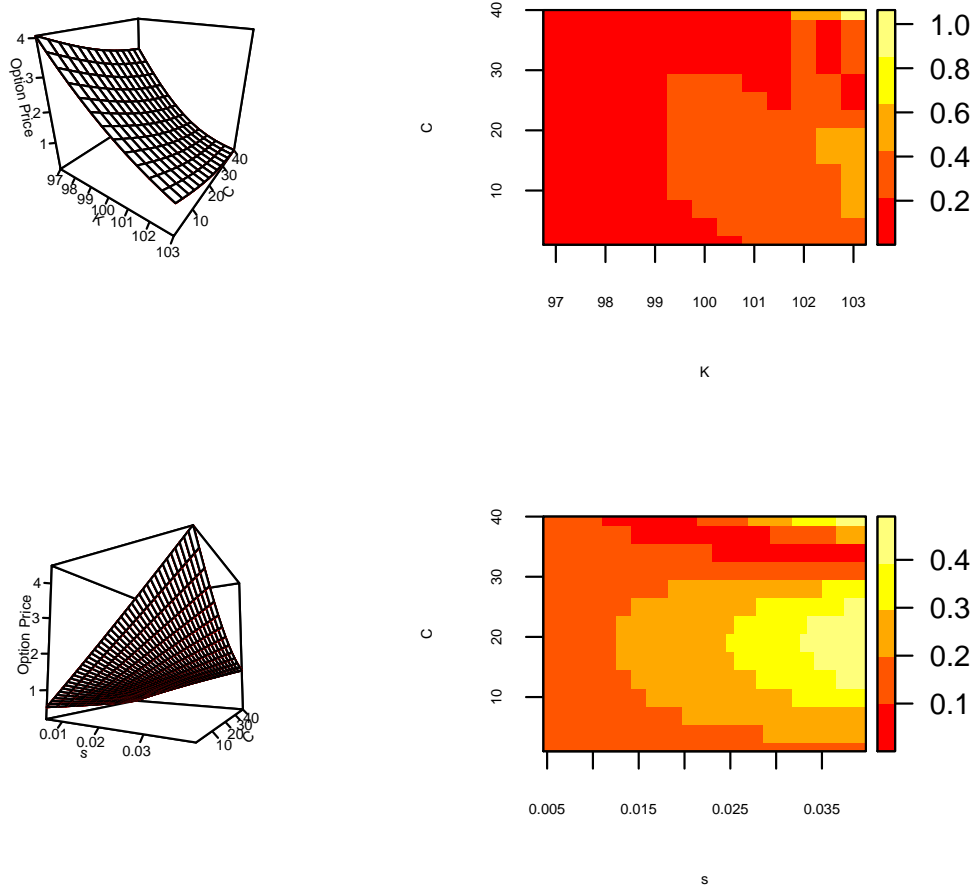}
\end{center}
\caption{Effect of strike and curvature on option price and effect of mortgage rate volatility and curvature on option price for shifted lognormal approximation and Monte Carlo. The heat map measures the percentage of relative differences between both approach.}\label{PriceSensitivity1}
\end{figure} 

The differences between the prices calculated by Monte Carlo and the close-form approximation are relatively small. In particular, when strike and volatility are growing the differences are generally increasing.

\subsection{Option Price Sensitivity on Lognormal approximation}
In this section we test the performance of parametric lognormal approximation given by Theorem \ref{parametriclognormal} under positively price distribution skewness. As the advantage of this analytically and parametric method compared to Monte Carlo is of course speed of computations; also this pricing methodology has very advantages into calibrations.

In order to ensure the assumptions, we start testing the lognormal hypothesis of empirical price distribution by performing Shapiro-Wilk Goodness-of-Fit lognormal test under low curvature levels. Table \ref{SWnormalitytest} contains the results and Figure \ref{graphical_test_lognormal} illustrates graphical test.

\begin{table}[H]
\centering
\begin{tabular}{|c|c|c|c|}
  \hline
  C & skew & W statistic & p-value  \\ 
  \hline
  0.5 & 0.123683 & 0.999249 & 0.616570   \\ 
  1 &   0.116235 &  0.999226 & 0.586158   \\ 
  2 &   0.101030 &  0.999164 & 0.508980   \\ 
  3 & 0.085431 &  0.999081 & 0.414202   \\ 
  4 & 0.069457 &  0.998977 & 0.311534  \\ 
  5 & 0.053131 &  0.998848 & 0.213987   \\ 
   6 & 0.036475 &  0.998694 & 0.133146   \\ 
   \hline
\end{tabular}
\caption{Summary statistics of Shapiro-Wilk test for lognormal assumptions of empirical price distribution.} 
\label{SWnormalitytest}
\end{table}

\begin{figure}[H]  
\begin{center}
\includegraphics[scale=0.8]{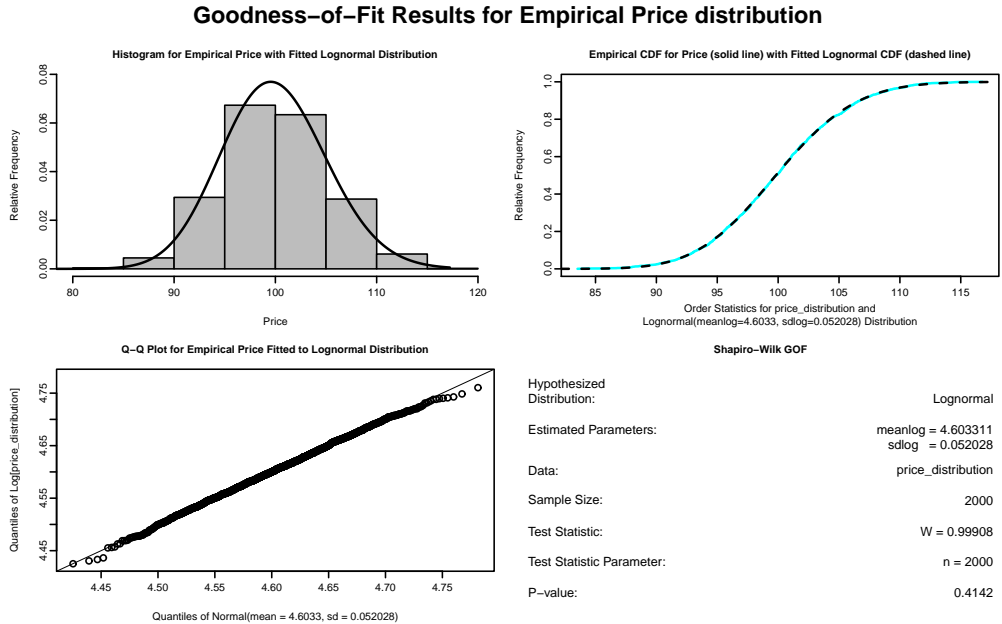}
\end{center}
\caption{Graphical test for lognormal empirical price approximation for $C=3$.}\label{graphical_test_lognormal}
\end{figure}

The test results showed that empirical price distribution may suppose lognormal for low curvature levels, and as expected when curvature increases p-values decreases evidencing the departure of lognormal assumptions.  

\bigskip

We will perform sensitivity test on option price and also for delta price. As before the reference values are computed by using Monte Carlo method. For reference delta, finite difference is used with 1bps (0.0001) on $P_0$ shock; more precisely the reference delta is computed by $\Delta_{MC}=\frac{\mathcal{C}_{MC}(P_0 + 0.0001) - \mathcal{C}_{MC}(P_0)}{0.0001}$. The strike $K$ is varied from 97 to 103, the current price $P_0$ is varied from 95 to 106, curvature is varied from 0.5 to 6 and all other parameters set to default values. Figure \ref{PriceSensitivity2} shows the results.

\begin{figure}[H]  
\begin{center}
\includegraphics[scale=1.1]{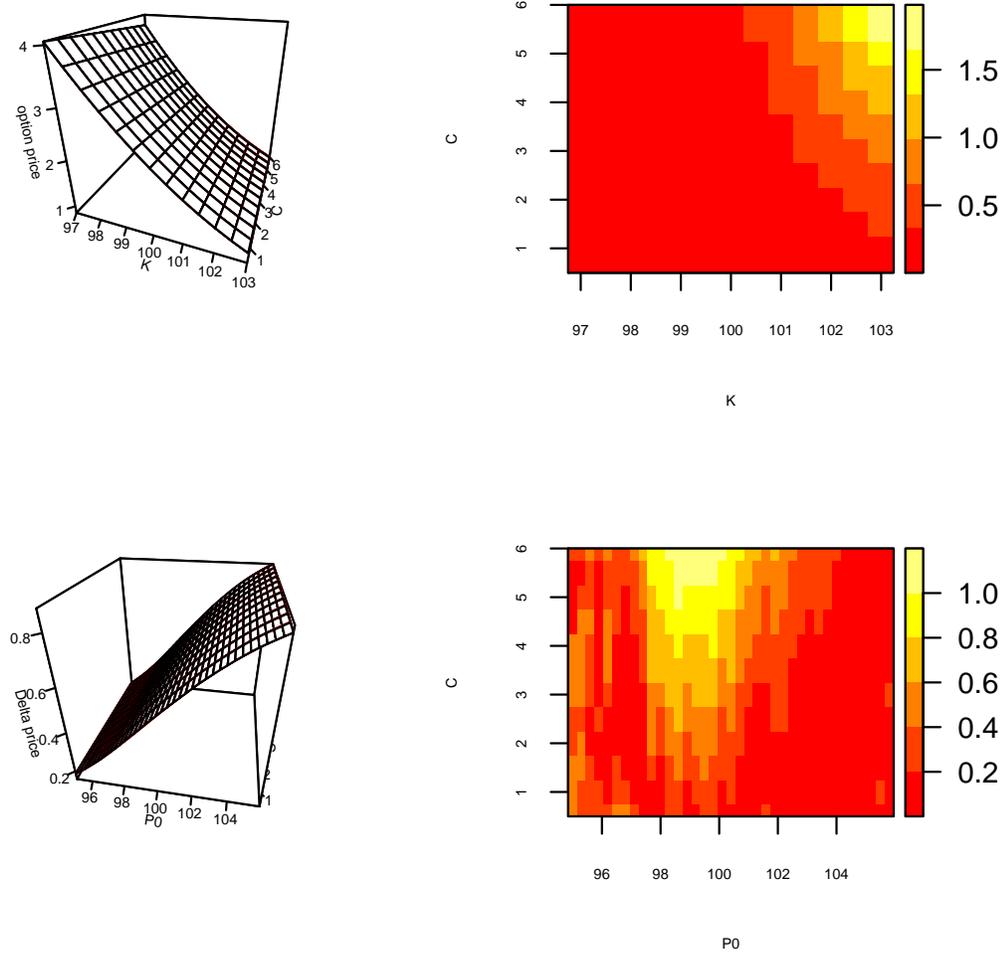}
\end{center}
\caption{Effect of strike and curvature on option price and effect of Price and curvature on Delta for lognormal approximation and Monte Carlo. The heat map measures the percentage of relative differences by performing $\frac{(\mathcal{C}_{LN} - \mathcal{C}_{MC})100}{\mathcal{C}_{MC}}$ for option price and $\frac{(\Delta_{LN} - \Delta_{MC})100}{\Delta_{MC}}$ for delta price.}\label{PriceSensitivity2}
\end{figure} 

The differences between the option prices and deltas calculated by Monte Carlo and the approach of lognormal approximation are relatively small. In particular, for option price it seems that the relative errors tend to increase when the option goes into in the money with higher curvature levels.

\section{Conclusion}

We have introduced a new approach for pricing mortgage option by approximating the mortgage price distribution by a family of lognormal distributions: regular or shifted lognormal. Hypothesis tests and graphical tests were performed and the results showed that the price distribution is close to shifted lognormal distributions and in particular for low curvature levels the price distribution is approximate by lognormal. These approximations allow us to make use of the Black-Scholes formulas for the European option price and option’s greeks. For low curvature levels, the lognormal distribution fits well to the price distribution, allowing the option price and the option's greeks to be approximated by means of Black-Scholes formulas in terms of the model parameters. Numerical simulations have shown that the option prices obtained by our analytic method approximate the prices resulting from Monte Carlo simulations remarkably well, and the delta performance of our method is also very good. The main advantage of our methodology over Monte Carlo approach is based on the computation speed and its possible applications to perform calibrations. 
   

%
%
%


\begin{thebibliography}{25}
\bibitem[Borovkova(2007)]{Borovkova(2007)} Borovkova, S., Permana, F. J., and Weide, H. V. (2007). A closed form approach to the valuation and hedging of basket and spread option. The Journal of Derivatives, 14(4), 8-24.
\bibitem[Hull(2006)]{Hull(2006)} Hull, J. (2006). Options, futures, and other derivatives. Upper Saddle River, NJ [ua]: Pearson Prentice Hall.
\bibitem[Kotz(2004)]{Kotz(2004)}Kotz, S., Balakrishnan, N., \& Johnson, N. L. (2004). Continuous multivariate distributions, Volume 1: Models and applications. John Wiley \& Sons.
\bibitem[Levy(1992)]{Levy(1992)} Levy, E. (1992). Pricing European average rate currency options. Journal of
International Money and Finance, 11(5):474–491. 
\bibitem[Safak(1994)]{Safak(1994)} Safak, A., \& Safak, M. (1994, June). Moments of the sum of correlated log-normal random variables. In Proceedings of IEEE Vehicular Technology Conference (VTC) (pp. 140-144). IEEE.
\bibitem[Prendergast(2003)]{Prendergast(2003)} Prendergast, J. R. (2003). The complexities of mortgage options. The Journal of Fixed Income, 12(4), 7-24.


\end{thebibliography}
\end{document}